\documentclass[10pt]{article}
\usepackage{paralist}
\usepackage{amssymb,color,amsmath,amsthm} 
\usepackage{cite}
\usepackage{epsfig} 
\newcommand{\nix}[1]{}

\newtheorem{theorem}{Theorem}
\newtheorem{lemma}[theorem]{Lemma}

\begin{document}

\title{A Note on Quantum Hamming Bound}
\author{Salah A. Aly \\ Department  of Computer Science,\\ Texas A\&M
University,\\
College Station, TX 77843-3112, USA \\
Email: salah@cs.tamu.edu\\
}
\date{}
\maketitle Proving the quantum Hamming bound for degenerate nonbinary
stabilizer codes has been an open problem for a decade.  In this note, I prove
this bound for double error-correcting degenerate stabilizer codes. Also, I
compute the maximum length of single and double error-correcting MDS stabilizer
codes over finite fields.

\section{Bounds on Quantum Codes}\label{sec:bounds}

Quantum stabilizer codes are a known class of quantum codes that can protect
quantum information against noise and decoherence. Stabilizer codes can be
constructed from self-orthogonal or dual-containing classical codes, see for
example\cite{calderbank98,ketkar06,gottesman97} and references therein. It is
desirable  to study upper and lower bounds on the minimum distance of classical
and quantum codes, so the computer search on the code parameters can be
minimized. It is  a well-known fact that Singleton and Hamming bounds hold for
classical codes~\cite{huffman03}. Also, upper and lower bounds on the
achievable minimum distance of quantum stabilizer codes are needed. Perhaps the
simplest upper bound is the quantum Singleton bound, also known as the
Knill-Laflamme bound. The binary version of the quantum Singleton bound was
first proved by Knill and Laflamme in~\cite{knill97}, see
also~\cite{ashikhmin99,ashikhmin00b}, and later generalized by Rains using
weight enumerators in~\cite{rains99}.

\begin{theorem}[Quantum Singleton Bound]\label{th:singleton}
An  $((n,K,d))_q$ stabilizer code with $K>1$ satisfies \begin{eqnarray} K\le
q^{n-2d+2}.\end{eqnarray}
\end{theorem}
Codes which meet the quantum Singleton bound with equality are called quantum
MDS codes. If we assume that $K=q^k$, then this bound can be stated as $k\leq
n-2d+2$. In~\cite{ketkar06} It has been shown that these codes cannot be
indefinitely long and showed that the maximal length of a $q$-ary quantum MDS
codes is upper bounded by $2q^2-2$. This could probably be tightened to
$q^2+2$. It would be interesting to find quantum MDS codes of length greater
than $q^2+2$ since it would disprove the MDS conjecture for classical codes
\cite{huffman03}. A related open question is regarding the construction of
codes with lengths between $q$ and $q^2-1$. At the moment there are no
analytical methods for constructing a quantum MDS code of arbitrary length in
this range (see \cite{grassl04} for some numerical results).

Another important bound for quantum codes is the quantum Hamming bound. The
quantum Hamming bound states (see~\cite{gottesman96,feng04}) that:
\begin{theorem}[Quantum Hamming Bound]\label{th:hamming}
Any pure $((n,K,d))_q$ stabilizer code satisfies
\begin{eqnarray}\sum_{i=0}^{\lfloor (d-1)/2\rfloor} \binom{n}{i}(q^2-1)^i \le
q^n/K.\end{eqnarray}
\end{theorem}
The previous quantum Hamming bound holds only for nondegenerate (pure) quantum
codes. However,   the degenerate (impure) quantum codes are particularly
interesting class of quantum codes because they can pack more quantum
information. In addition, the errors of small weights do not need active error
correction strategies.

So far no degenerate quantum code has been found that beats this bound.
Gottesman showed that impure single and double error-correcting binary quantum
codes cannot beat the quantum Hamming bound~\cite{gottesman97}. It is proved
in~\cite{ketkar06} that Hamming bound holds for  quantum stabilizer codes with
distance $d=3$.

In general, does Hamming bound exist for any distance $d$ in $((n,K,d))_q$
stabilizer codes? This has been an open question for a decade.  In this note we
prove Hamming bound for double error-correcting stabilizer codes with distance
$d=5$ and also give a sketch to prove it for general distance $d$.

\section{Quantum Hamming Bound Holds for Distance $d=5$}
There have been several approaches to prove bounds on the quantum code
parameters. In \cite{ashikhmin99} Ashikhmin and Litsyn derived many bounds for
quantum codes by extending a novel method originally introduced by
Delsarte~\cite{delsarte72} for classical codes. Using this method they proved
the binary versions of Theorems~\ref{th:singleton},\ref{th:hamming}. We use
this method to show that the Hamming bound holds for all double
error-correcting quantum codes. See \cite{ketkar06} for a similar result for
single error-correcting codes. But first we need Theorem~\ref{th:lp2} and the
Krawtchouk polynomial of degree $j$ in the variable $x$,
\begin{eqnarray}
K_j(x) = \sum_{s=0}^j (-1)^s(q^2-1)^{j-s} {x \choose s}{n-x \choose
j-s}.\end{eqnarray}

\begin{theorem}\label{th:lp2}
Let $Q$ be an $((n,K,d))_q$ stabilizer code of dimension $K>1$. Suppose that
$S$ is a nonempty subset of $\{0,\dots,d-1\}$ and $N=\{0,\dots,n\}$.  Let

\begin{eqnarray} f(x)= \sum_{i=0}^n f_i K_i(x)\end{eqnarray} be a polynomial satisfying the conditions
\begin{enumerate}
\item[i)] $f_x> 0$ for all $x$ in $S$, and $f_x\ge 0$ otherwise;
\item[ii)] $f(x)\le 0$ for all $x$ in $N\setminus S$.
\end{enumerate}
Then \begin{eqnarray} K \le \frac{1}{q^n}\max_{x\in S}
\frac{f(x)}{f_x}.\end{eqnarray}
\end{theorem}
\begin{proof}{}
See~\cite{ketkar06}.
\end{proof}
We demonstrate usefulness of the previous Theorem by showing that quantum
Hamming bound holds for impure codes  when $d=5$.


\begin{lemma}[Quantum Hamming Bound]\label{th:hamming2}
An $((n,K,5))_q$ stabilizer code with $K>1$ satisfies

\begin{eqnarray}K\leq q^{n}\big/\Big(n(n-1)(q^2-1)^2/2+n(q^2-1)+1\Big).\end{eqnarray}
\end{lemma}
\begin{proof}{}
 Let $f(x)=\sum_{j=0}^n f_j K_j(x)$, where $f_x=(\sum_{j=0}^{e}K_j(x))^2$,  $S = \{0,1,\ldots,4\}$ and N=\{0,1,\ldots,n\}. Calculating $f(x)$ and $f_x$
gives us
\begin{eqnarray*}
  f_0 &=& (1+n(q^2-1)+n(n-1)(q^2-1)^2/2)^2\\
  f_1 &=& \frac{1}{4} (n-1)^2 (n-2)^2 (q^2-1)^4  \\
  f_2 &=& (\frac{1}{2} (n-3) (n-2) (q^2-1)^2 -(n-2) (q^2-1))^2\\
  f_3 &=& (1-2 (n-3) (q^2-1)+\frac{1}{2} (n-4) (n-3) (q^2-1)^2)^2 \\
  f_4 &=& (3-3 (n-4) (q^2-1)+\frac{1}{2} (n-5) (n-4) (q^2-1)^2)^2\\
  \mbox{ and, }\\
    f(0) &=& q^{2 n} (1+n (q^2-1)+\frac{1}{2} (n-1) n (q^2-1)^2)\\
    f(1) &=& q^{2 n} (q^2+2 (n-1) (q^2-1)+(n-1) (q^2-2) (q^2-1)) \\
    f(2) &=& q^{2 n} (4+4 (q^2-2)+(q^2-2)^2+2 (n-2) (q^2-1)) \\
    f(3) &=& q^{2 n} (6+6 (q^2-2)) \\
    f(4) &=& 6 q^{2 n}.
\end{eqnarray*}
Clearly $f_x>0$ for all $x \in S$ . Also, $f(x) \leq 0$ for all $x \in N
\backslash S$ since the binomial coefficients for the negative values are zero.
The Hamming bound is given by

\begin{eqnarray} K \leq q^{-n} \max_{s \in S} \frac{f(x)}{f_x}\end{eqnarray}

 So, there are four different comparisons where $f(0)/f_0 \geq f(x)/f_x$, for $x=1,2,3,4$. We find
a lower bound for $n$ that holds for all values of $q$. From
Lemmas~\ref{lem:compare0},\ref{lem:compare1},\ref{lem:compare2}, and
\ref{lem:compare3}, shown below, for $n \geq 7$ it follows
 that \begin{eqnarray} \max
\{f(0)/f_0,f(1)/f_1,f(2)/f_2,f(3)/f_3,f(4)/f_4\} = f(0)/f_0
\end{eqnarray}

\end{proof}
While the above method is a general method to prove Hamming bound for impure
quantum codes,  the number of terms increases with a large minimum distance. It
becomes difficult to find the true bound using this method. However, one can
derive more consequences from Theorem~\ref{th:lp2}; see, for
instance,~\cite{ashikhmin99,ashikhmin00b,levenshtein95,mceliece77}.

\begin{lemma}\label{lem:compare0}
The inequality $f(0)/f_0 \geq f(1)/f_1$ holds for $n  \geq 6$ and $q \geq 2$.
\end{lemma}
\begin{proof}{}
Let $f(0)/f_0 \geq f(1)/f_1$ then\\
\begin{eqnarray*}   
\frac{1}{1+n(q^2-1)+n(n-1)(q^2-1)^2/2} &\geq&
\frac{4q^2((n-1)(q^2-1)+1)}{(n-1)^2(n-2)^2(q^2-1)^4} \nonumber
\end{eqnarray*}
\begin{eqnarray*}
(n-1)^2(n-2)^2(q^2-1)^4   &\geq& (1+n(q^2-1)\\
&+& \frac{n(n-1)}{2}(q^2-1)^2)(4q^2((n-1)(q^2-1)+1)) \nonumber
\end{eqnarray*}
in the left side $(n-1)$ approximates to $(n-2)$. Also, in the right side
$(n-2)$ and $(n-1)$ approximate to $(n)$. So,
\begin{eqnarray*}
(n-2)^4(q^2-1)^4 &\geq& 4(1+n(q^2-1)+\frac{n^2}{2}(q^2-1)^2)(q^2(q^2-1)(n-1)+1)
\nonumber
\end{eqnarray*}
divide both sides by $(q^2-1)^2(q^2-1)^2$ and approximate $\frac{1}{q^2-1} \leq
1$, we find that
\begin{eqnarray*}
(n-2)^4 \geq 8(1+n+\frac{n^2}{2})(n-1) \nonumber
\end{eqnarray*}
by approximating both sides, the final result is   $(n-2) \geq 4 $ or
\begin{eqnarray}
 n \geq 6 \nonumber
\end{eqnarray}
\end{proof}
\begin{lemma}\label{lem:compare1}
The inequality $f(0)/f_0 \geq f(2)/f_2$ holds for  $n \geq  7 $  and $q \geq
2$.
\end{lemma}
\begin{proof}{}
Let \begin{eqnarray}   
\frac{q^{2n}}{1+n(q^2-1)+n(n-1)(q^2-1)^2/2} \geq \frac{q^{2n}(q^4+2(n-2)(q^2-1)
)}{(-(n-2)(q^2-1)+(n-3)(n-2)(q^2-1)^2/2)^2} \nonumber
\end{eqnarray}
by simplifying both sides\\
\begin{eqnarray}
(-(n-2)(q^2-1)+(n-3)(n-2)(q^2-1)^2/2)^2 &\geq &\nonumber \\
(q^4+2(n-2)(q^2-1))(1+n(q^2-1)+n(n-1)(q^2-1)^2/2) \nonumber
\end{eqnarray}
Simplifying  L.H.S, $(n-2)$ to $(n-3)$ then
\begin{eqnarray}
(q^2-1)^4((n-3)^2/2 -(n-2))^2 &\geq& \nonumber
\\(q^4+2(n-2)(q^2-1))(1+n(q^2-1)+n(n-1)(q^2-1)^2/2) \nonumber
\end{eqnarray}
by simplifying both sides
\begin{eqnarray}
((n-3)^2/2 -(n-2))^2 &\geq& (\frac{q^2}{(q^2-1)^2} + \frac{2(n-2)}{q^2-1} )(1+n+n(n-1)/2) \nonumber\\
((n-3)^2/2 -(n-2))^2 & \geq& 2(2n+1)(n^2+2n+2) \nonumber\\
(n-3)^2((n-3)/2 -1)^2 &\geq& 2(2n+1)((n+1)^2+1) \nonumber\\
(n-5)^2/4 &\geq& 2(2n+1)\nonumber \\
n &\geq& 7 \nonumber
\end{eqnarray}
\end{proof}

\begin{lemma}\label{lem:compare2}
The inequality $f(0)/f_0 \geq f(3)/f_3$ holds for  $n \geq  7$  and $q \geq 2$.
\end{lemma}
\begin{proof}{}
Let
\begin{eqnarray}
\frac{q^{2n}}{1+n(q^2-1)+n(n-1)(q^2-1)^2/2} \geq
\frac{6q^{2n}(q^2-1)}{(1-2(n-3)(q^2-1)+(n-3)(n-4)(q^2-1)^2/2)^2} \nonumber
\end{eqnarray}
by simplification
\begin{eqnarray}
(1-2(n-3)(q^2-1)+(n-3)(n-4)(q^2-1)^2/2)^2 &\geq& \nonumber \\ 6(q^2-1) (1+n(q^2-1)+n(n-1)(q^2-1)^2/2) \nonumber\\
 \frac{((n-4)^2(q^2-1)^2-4(n-4)(q^2-1)+2)^2}{4} \nonumber &\geq& \\
 6(q^2-1)(1+n(q^2-1)+n(n-1)(q^2-1)^2/2) \nonumber
\end{eqnarray}
by approximation to $(q^2-1)$
\begin{eqnarray}
\frac{(q^2-1)^4}{4}((n-5)^4 &\geq& 3(q^2-1)^3(2+2n+n^2) \nonumber\\
(n-5)^4 &\geq& 4( 2+2n+n^2 ) \nonumber\\
(n-5)^2 &\geq& 2 \nonumber\\
n &\geq& 7 \nonumber
\end{eqnarray}
\end{proof}

\begin{lemma}\label{lem:compare3}
The inequality $f(0)/f_0 \geq f(4)/f_4$ holds for   $n \geq 7 $ and $q \geq 2$.
\end{lemma}
\begin{proof}{}
Let
\begin{eqnarray}   
\frac{q^{2n}}{1+n(q^2-1)+n(n-1)(q^2-1)^2/2} \nonumber \\ \geq
\frac{6q^{2n}}{(3-3(n-4)(q^2-1)+(n-4)(n-5)(q^2-1)^2 /2)^2} \nonumber
\end{eqnarray}
divide by $q^{2n}$ and simplifying
\begin{eqnarray}
 (3-3(n-4)(q^2-1)+(n-4)(n-5)(q^2-1)^2 /2)^2 &\geq& \nonumber \\
 6(1+n(q^2-1)+n(n-1)(q^2-1)^2/2) \nonumber
\end{eqnarray}
then by approximating $ (n-4)$ to $(n-5)$ in L.H.S and $(n-1)$ to $4n$ in R.H.S, we can find that\\
\begin{eqnarray}   
(-3(n-4)(q^2-1)+(n-5)^2(q^2-1)^2 /2)^2 &\geq& 6(1+n(q^2-1)+n^2(q^2-1)^2/2)
\nonumber
\end{eqnarray}
\begin{eqnarray}   
((n-5)^2(q^2-1)+(n-5)^2(q^2-1)^2 /2)^2 &\geq& 6(1+n(q^2-1)+n^2(q^2-1)^2/2)
\nonumber
\end{eqnarray}
\begin{eqnarray}   
(q^2-1)^4((n-5)^2+(n-5)^2 /2)^2 &\geq& 6(1+n(q^2-1)+n^2(q^2-1)^2/2) \nonumber
\end{eqnarray}
dividing both sides by $(q^2-1)^4$ and simplifying
\begin{eqnarray}   
(9/4)(n-5)^4 &\geq& \frac{6(1+n(q^2-1)+n^2(q^2-1)^2/2)}{(q^2-1)^4} \nonumber
\end{eqnarray}
\begin{eqnarray}   
(9/4)(n-5)^4 &\geq& 6(1+n+n^2)\nonumber\\
 n  &\geq&  7 \nonumber
\end{eqnarray}
\end{proof}

Since it is not known if the quantum Hamming bound holds for degenerate
nonbinary quantum codes, it would be interesting to find degenerate quantum
codes that either meet or beat the quantum Hamming bound. This is obviously a
challenging open research problem.

\section{Maximal Length of MDS Codes}
 In this section we derive some results on the maximal length of
 single and double error-correcting quantum MDS codes. These
bounds hold for all additive quantum codes.

\subsection{Maximal Length Single Error-correcting MDS Codes}
\begin{lemma}\label{max1}
The maximal length of single error correcting additive quantum MDS codes is
given by $q^2+1$.
\end{lemma}
\begin{proof}
We know that the quantum Hamming bound holds for $K>1$ for $d=3$, so
\begin{eqnarray}
K &\leq & \frac{q^n}{1+n(q^2-1)}
\end{eqnarray}
If the Hamming bound is tighter than the Singleton bound for any $((n,K,3))_q$
quantum code, then it means that MDS codes cannot exist for that set of $n,K$.
This occurs when
\begin{eqnarray}
q^{n-2d+2}=q^{n-4 } &\geq& \frac{q^n}{1+n(q^2-1)}\nonumber\\
1+n(q^2-1) &\geq& q^4\nonumber\\
n&\geq& q^2+1
\end{eqnarray}
Thus there exist no single error correcting quantum MDS codes for $n>q^2+1$.
\end{proof}
\subsection{Upper Bound on the Maximal Length of Double Error-correcting MDS
Codes}
\begin{lemma}\label{max2}
 The maximal length of double error-correcting quantum MDS
codes is upper bounded by:\\
\begin{eqnarray}
n \leq \frac{(q^2-3)+ \sqrt{( (q^2 -3) +8(q^8 -1))}}{2(q^2-1)}
\end{eqnarray}
\end{lemma}
\begin{proof}
It is known that the Hamming bound for $d = 5$ is given by:
\begin{eqnarray}   
K &\leq & \frac{q^n}{1+n(q^2-1)+n(n-1)(q^2-1)^2/2}
\end{eqnarray}
If the Hamming bound is tighter than the Singleton bound for any $((n,K,5))_q$
quantum code, then it means that MDS codes cannot exist for that set of code
parameters. By simple computation, we find that
\begin{eqnarray*}
q^{n-2d+2}=q^{n-8 } &\geq& \frac{q^n}{1+n(q^2-1)+ \frac{n(n-1)}{2}(q^2-1)^2}
\end{eqnarray*}
\begin{eqnarray}
n^2(q^2-1)^2 - n(q^2-1)(q^2-3)-2(q^8-1) \geq  0
\end{eqnarray}
So, the quadratic equation of $n$ has two real solutions. This inequality holds
for
\begin{eqnarray}
n \geq \frac{(q^2-3)-\sqrt{(q^2-3)^2+8(q^8-1)}}{2(q^2-1)}\\
n \leq \frac{(q^2-3)+\sqrt{(q^2-3)^2+8(q^8-1)}}{2(q^2-1)}
\end{eqnarray}

Only the positive solution  for $n$ is valid. So, the maximal length of double error-correcting MDS code is upper bounded by \\
\begin{eqnarray}
n \leq \frac{(q^2-3)+\sqrt{(q^2-3)^2+8(q^8-1)}}{2(q^2-1)}
\end{eqnarray}

\end{proof}

\def\cprime{$'$}


\begin{thebibliography}{10}

\bibitem{ashikhmin99}
A.~Ashikhmin and S.~Litsyn.
\newblock Upper bounds on the size of quantum codes.
\newblock {\em IEEE Trans. Inform. Theory}, 45(4):1206--1215, 1999.

\bibitem{ashikhmin00b}
A.E. Ashikhmin, A.M. Barg, E.~Knill, and S.N. Litsyn.
\newblock Quantum error detection {II}: Bounds.
\newblock {\em IEEE Trans. on Information Theory}, 46(3):789--800, 2000.

\bibitem{calderbank98}
A.R. Calderbank, E.M. Rains, P.W. Shor, and N.J.A. Sloane.
\newblock Quantum error correction via codes over {GF}(4).
\newblock {\em IEEE Trans. Inform. Theory}, 44:1369--1387, 1998.

\bibitem{delsarte1973}
P.~Delsarte.
\newblock Four fundamental parameters of a code and their combinatorial
  significance.
\newblock {\em Information and Control , 23(5):407-438, December 1973}.

\bibitem{delsarte72}
P.~Delsarte.
\newblock Bounds for unrestricted codes by linear programming.
\newblock {\em Philips Res. Reports}, 27:272--289, 1972.

\bibitem{feng04}
K.~Feng and Z.~Ma.
\newblock A finite {G}ilbert-{V}arshamov bound for pure stabilizer quantum
  codes.
\newblock {\em IEEE Trans. Inform. Theory}, 50(12):3323--3325, 2004.

\bibitem{gottesman96}
D.~Gottesman.
\newblock A class of quantum error-correcting codes saturating the quantum
  {H}amming bound.
\newblock {\em Phys. Rev. A}, 54:1862--1868, 1996.

\bibitem{gottesman97}
D.~Gottesman.
\newblock Stabilizer codes and quantum error correction.
\newblock Caltech Ph. D. Thesis, eprint: quant-ph/9705052, 1997.

\bibitem{grassl04}
M.~Grassl, T.~Beth, and M.~R{\"o}tteler.
\newblock On optimal quantum codes.
\newblock {\em Internat. J. Quantum Information}, 2(1):757--775, 2004.

\bibitem{huffman03}
W.~C. Huffman and V.~Pless.
\newblock {\em Fundamentals of Error-Correcting Codes}.
\newblock University Press, Cambridge, 2003.

\bibitem{ketkar06}
A.~Ketkar, A.~Klappenecker, S.~Kumar, and P.K. Sarvepalli.
\newblock Nonbinary stabilizer codes over finite fields.
\newblock {\em IEEE Trans. Inform. Theory}, 52(11):4892 -- 4914, 2006.

\bibitem{knill97}
E.~Knill and R.~Laflamme.
\newblock {A theory of quantum error--correcting codes}.
\newblock {\em Physical Review~A}, 55(2):900--911, 1997.

\bibitem{levenshtein95}
V.I. Levenshtein.
\newblock Krawtchouk polynomials and universal bounds for codes and designs in
  {H}amming spaces.
\newblock {\em IEEE Trans. Inform. Theory}, 41(5):1303--1321, 1995.

\bibitem{Lint99}
J.H~Van Lint.
\newblock Introduction to coding theory.
\newblock {\em Third Edition, Springer-Verlag 1999}.

\bibitem{mceliece77}
R.J. McEliece, E.R. Rodemich, jr.~H. Rumsey, and L.R. Welch.
\newblock New upper bounds on the rate of a code via the
  {D}elsarte-{M}ac{W}illiams inequalities.
\newblock {\em IEEE Trans. Inform. Theory}, 23(2):157, 1977.

\bibitem{rains99}
E.M. Rains.
\newblock Nonbinary quantum codes.
\newblock {\em IEEE Trans. Inform. Theory}, 45:1827--1832, 1999.

\end{thebibliography}

\newpage
\section{Appendix}

\section*{An  Approach (Sketch) to Prove Hamming Bound for Degenerate Nonbinary Stabilizer Codes with Minimum Distance $d$}
One way to prove the quantum Hamming bound for impure nonbinary stabilizer
codes with $d \leq (n-k+2)/2$ is to expand $f(x)/f_x$ in terms of Krawtchouk
polynomials. Let $f(x) = \sum_{j=0}^n f_j K_j(x)$ and $f_x =
(\sum_{i=0}^{e}K_i(x))^2 $. The Krawtchouk polynomial of degree e in the
variables x and q is given by
\begin{eqnarray}
K_{e}(q,x) =  \sum_{j=0}^e (-1)^{j} (q^{2} -1)^{e-j} \begin{pmatrix}   x \\   j
\\ \end{pmatrix} \begin{pmatrix}   n-x \\   e-j \\ \end{pmatrix}
\end{eqnarray}

\begin{theorem}\label{maxbound}
Let $Q$ be an $((n,K,d))_q$ stabilizer code of dimension $K>1$. Suppose that
$S$ is a nonempty subset of $\{0,1,...,d-1\}$ and $N =\{0,1,...,n\}$. Let
\[f(x) = \sum_{i=0}^{n} f_i K_i(x) \] be a polynomial satisfying the
conditions:
\begin{itemize}
    \item[i)] $f_x > 0$ for all $x \in S$, and $f_x \geq 0$ otherwise;
\item[ii)] $f(x) \leq 0$ for all $x \in N\backslash S$.
\end{itemize}
Then \[ K \leq \frac{1}{q^{n}} \max_{x \in S} \frac{f(x)}{f_x}. \]
\end{theorem}
Notice that $ f(x) = \sum_{i=0}^{n} f_i K_i(x)   $ can be written as $f_i  =  q^{-2n} \sum_{x=0}^{n}f(x) K_x(i)$.\\


\begin{lemma}[Sketch] \label{hamming-main23} Let $Q$ be an $((n,K,d))_q$ stabilizer code of dimension $k\geq 1$. Suppose that $S$ is a non-empty subset of
\{0,1,2,....,2e\}, where $e = \lfloor\frac{d-1}{2}\rfloor$.
The Hamming bound is given by $K \leq q^{-n} \max_{x \in S} \frac{f(x)}{f_x}$ equals to  \\
\begin{eqnarray}K \leq \frac{q^{n}}{\sum_{i=0}^{e}\left(%
\begin{array}{c}
  n\\
  i \\
\end{array}%
\right)(q^{2}-1)^{i}}
\end{eqnarray}
If and only if $f(0)/f_0$ is the maximum value for $d\geq 3$ and $n \geq n_0$.
\end{lemma}

\begin{proof} In this proof, we propose $f_x$ satisfying Theorem~\ref{maxbound}.
Let $f_x = \Big(\sum_{i=0}^{e}K_i(x)\Big)^2 $ and $f(x) = \sum_{j=0}^n f_j
K_j(q,x)$.\\\\
\begin{eqnarray}
\frac{f(x)}{f_x}  =  \frac{ \sum_{j=0}^n f_j K_j(q,x)  }{f_x  }
\end{eqnarray}
And our goal is to find $\max \{f(0)/f_0,f(1)/f_1,..., f(d-1)/f_{d-1} \}$ that may equal to $f(0)/f_0$.\\

 Now, for $x=0$, we find that
\begin{eqnarray}\label{eq38}
\frac{f(0)}{f_0}  &=&  \frac{\sum_{j=0}^n f_j K_j(0)}{f_0} \nonumber \\ &=&
K_0(q,0)+\frac{f_1 K_1(q,0)}{f_0}+......+\frac{f_n K_n(q,0)}{f_0}
\end{eqnarray}
or
\begin{eqnarray}
\frac{f(0)}{f_0}  =\frac{\sum_{j=0}^n \big(\sum_{i=0}^{e}K_i(j)\big)^2
K_j(0)}{\big(\sum_{i=0}^{e}K_i(0 )\big)^2}\nonumber
\end{eqnarray}

 and for any other value of $y \in \{1,2,...,d-1 \}$ , we find that
\begin{eqnarray}\label{eq39}
\frac{f(y)}{f_y}  &=&  \frac{ \sum_{j=0}^n f_j K_j(y)  }{f_y   }  \nonumber \\
&=& \frac{f_0 K_0(q,y)}{f_y}+\frac{f_1 K_1(q,y)}{f_y}+......+\frac{f_n
K_n(q,y)}{f_y}
\end{eqnarray}
or
\begin{eqnarray}
\frac{f(y)}{f_y} = \frac{\sum_{j=0}^n \big(\sum_{i=0}^{e}K_i(j) \big)^2
K_j(y)}{\big(\sum_{i=0}^{e}K_i(y)\big)^2}\nonumber
\end{eqnarray}

From \ref{eq38} and \ref{eq39}, simply  we need to show that
\begin{eqnarray}
\frac{f(0)}{f_0} - \frac{f(y)}{f_y} \geq 0
\end{eqnarray}

\begin{eqnarray} \label{eq42}
 \frac{f(0)}{f_0} - \frac{f(y)}{f_y}&=&\frac{\sum_{j=0}^n \big(\sum_{i=0}^{e}K_i(j)\big)^2 K_j(0)}{\big(\sum_{i=0}^{e}K_i(0 )\big)^2} -\frac{\sum_{j=0}^n \big(\sum_{i=0}^{e}K_i(j)\big)^2 K_j(y)}{\big(\sum_{i=0}^{e}K_i(y)\big)^2}
  \nonumber \\&=& \sum_{j=0}^n \left(\big(\sum_{i=0}^{e}K_i(j)\big)^2 \Big(\frac{K_j(0)}{\big(\sum_{i=0}^{e}K_i(0 )\big)^2}-\frac{K_j(y)}{\big(\sum_{i=0}^{e}K_i(y
  )\big)^2}\Big)\right)
\nonumber \\ &=& \sum_{j=0}^n \left( \frac{f_j K_j(q,0)}{f_0}-\frac{f_j
K_j(q,y)}{f_y} \right)
\end{eqnarray}

in the previous equation, $f_j > 0$ and $f_y > 0$, so, if we prove  that
\begin{eqnarray} \label{equ-seekconstant}
 \frac{f_j K_j(q,0)}{f_0}-\frac{f_j K_j(q,y)}{f_y} \geq 0,
\end{eqnarray}
then the claim holds. As shown in~\cite{delsarte1973},~\cite{Lint99}, we seek a constant value for the left side in~\ref{equ-seekconstant}, so, multiplying both sides by $K_e(i)$\\

\begin{eqnarray}
\frac{K_e(i)K_i(q,0)}{f_0}-\frac{K_e(i)K_i(q,y)}{f_y} \geq 0
\end{eqnarray}
and take $\sum_{i=0}^{n}$
\begin{eqnarray}
\sum_{i=0}^{n} \left(\frac{K_e(i)K_i(q,0)}{f_0}-\frac{K_e(i)K_i(q,y)}{f_y}\right) \geq 0 \nonumber \\
\frac{\sum_{i=0}^{n}K_e(i)K_i(q,0)}{f_0}-\frac{\sum_{i=0}^{n}K_e(i)K_i(q,y)}{f_y}
\geq 0
\end{eqnarray}
from \cite{Lint99}, given that $\sum_{i=0}^{n}K_e(i)K_i(q,j) = q^n \delta _{ej}$, by substitution, \\

\begin{eqnarray}
\frac{q^n\delta_{e0}}{f_0}-\frac{q^n\delta_{ey}}{f_y} \geq 0 \nonumber\\
\frac{\delta_{e0}}{f_0}-\frac{\delta_{ey}}{f_y} \geq 0
\end{eqnarray}
Now, $\delta_{e0} =1$, and $\delta_{ey}= 1$ or $0$; and obviously $f_y \geq
f_0$.  So, if $y=e \Longrightarrow \delta_{e0}/f_0 \geq 0$, and similarly,
$\delta_{ey} =1 \Longrightarrow f_y -f_0 \geq 0$.

\end{proof}

\end{document}